\documentclass[11pt]{article}
\usepackage{amsbsy}
\usepackage{amsfonts}
\usepackage{amsmath}
\usepackage{amsthm}
\usepackage{amssymb}
\usepackage{apacite}
\usepackage{setspace}
\usepackage{graphicx}
\usepackage{bm,multicol}
\usepackage[english]{babel}
\usepackage{natbib}
\usepackage[T1]{fontenc}
\usepackage[utf8]{inputenc}
\usepackage{authblk}
\usepackage{xcolor}
\usepackage{natbib}
\usepackage[title]{appendix}

\newtheorem{lemma}{Lemma}

\newtheorem{theorem}{Theorem}

\newtheorem{definition}{Definition}

\newtheorem{remark}{Remark}

\doublespace
\newcommand{\newc}{\newcommand}
\newc{\N}{\mbox{N}}
\newc{\1}{\bf{1}}
\def\signed #1{{\leavevmode\unskip\nobreak\hfil\penalty50\hskip2em
  \hbox{}\nobreak\hfil(#1)%
  \parfillskip=0pt \finalhyphendemerits=0 \endgraf}}

\newsavebox\mybox
\newenvironment{aquote}[1]
  {\savebox\mybox{#1}\begin{quote}}
  {\signed{\usebox\mybox}\end{quote}}

\begin{document}
\title{Converting P-Values in Adaptive Robust Lower Bounds of Posterior Probabilities to
increase the reproducible  Scientific ``Findings''}
\author{L. R. Pericchi \& M.E. Perez}
\affil{University of Puerto Rico, R\'{\i}o Piedras Campus}
\date{November 16, 2017}

\maketitle
\begin{abstract}
We put forward a novel calibration of $p$ values, the {\it ``Adaptive Robust Lower Bound'' (ARLB)} which maps $p$ values into approximations of posterior probabilities taking into account the effect of sample sizes. We build on the Robust Lower Bound proposed by \citet{SBB2001}, but we incorporate a simple power of the sample size to make it adaptive to different amounts of data.\\
We present several illustrations from where it is apparent that the ARLB closely approximates exact Bayes Factors. In particular, it has the same asymptotics as posterior probabilities but avoiding the problems of ``Bayesian Information Criterion" (BIC) for small samples relative to the number of parameters.\\
We prove that the ARLB is consistent as the sample size grows, and that it is information consistent (Berger and Pericchi, 2001)
for the canonical Normal case, but with methods that are keen to be generalized. So ARLB also avoids the problems of certain conjugate priors as g-priors.\\
In summary, this is a novel criterion easy to apply, as it only requires a real $p$ value, a sample size and parameter dimensionality. This method is intended to aid the practitioners, who are increasingly aware of the lack of reproducibility of traditional hypothesis testing ``findings" but at the same time, lack of concrete simple alternatives. Here is one.
\end{abstract}

\section{The Crisis of $p\mbox{ value}<0.05$}
\begin{aquote}{Bradley Efron, 2010}
\textit{Fishers's scale of evidence, particularly the $\alpha=0.05$ threshold, has been used in literally millions of serious scientific studies, and takes a good claim to being the 20th century's most influential piece of applied mathematics.}
\end{aquote}

Bayesian literature have been criticizing for several decades the implementation of hypothesis testing with fixed significance levels, and in particular the use of the scale $p \mbox{ value}<0.05$. That discussion was mostly regarded as a philosophical issue about the wrong interpretation of $p$-values as probabilities of the null hypothesis. However, the crisis of Fisher's scale of evidence exploded when scientific researchers, largely outside departments of Mathematics and Statistics, began reporting very low rates of reproducible scientific presumed findings. Many papers and opinions have been written on this subject, and we will mention just a few of them:
\begin{itemize}
\item In 2005, John Ioannidis publish a paper in PLOS Medicine entitled ``Why Most Published Research Findings Are False'' (Ioannidis, 2005)
\item Sir David Cox: ``Statistics is also about science and $p<0.05$ is seen as the passport for publication, and most if not all statisticians would take a rather disapproving view of it to put it mildly'' \citep{Cox2015}.
\item In 2015, the Basic and Applied Social Psychology Journal banned the use of significance testing, so $p<0.05$ will not anymore be a ``passport for publication''. \citep{Trafimow2015}
\item In March 2016, the American Statistical Association publishes a Statement on Statistical Significance and P-Values, establishing some general principles for the use and interpretation of p-values, principles on which all Statisticians could agree. (Wasserstein and Lazar, 2016)
\item In September 2017, a group of more than 70 researchers in Statistics published a paper asking for "redefining statistical significance". They propose to change the default $p$-value threshold for statistical significance from 0.05 to 0.005 for claims of new discoveries \citep{Benjamin2017}.
\end{itemize}

It is then clear that obtaining a $p$-value lower than $0.05$ does not open the doors for publication as before, and we Statisticians must provide alternatives to Scientists.

%
%
%
%

\subsection{How the old passport for publication should be replaced?}

The natural alternative to  Null Hypothesis Significance (NHST) methods, then,  would be using exact posterior probabilities for the hypothesis as the decision tool. Unfortunately, these probabilities are rarely available to scientists. Previous efforts in this direction include Schwarz's {\em Bayesian Information Criterion} (BIC) \citep{schwarz1978} as an approximation to a Bayes factor for large sample sizes. The BIC can be easily calculated in many problems and it avoids the formal specification of prior distributions; nevertheless, it produces results far from the actual Bayes Factor even for very large sample sizes \citep{KR1995}
 
 In contrast, tools for calculating p-values are widely available. This fact suggests finding bridges between  p-values and posterior probabilities of hypothesis, thus improving the decision making process.


One such bridge is presented in \citet{PP2014}.  A classical problem in the theory of statistics has been how p-values be corrected from their dependence on the sample size. P\'erez and  Pericchi find an answer by making an alternative question: How can  p-values be calibrated so that to have  the same asymptotic behavior as posterior probabilities? . The resulting proposal is  the {\em Adaptive $\alpha$ Level},

 \begin{equation}
\alpha_{n^*}(q)=\frac{[\chi^2_{\alpha}(q)+q\log(n^*)]^{\frac{q}{2}-1}}{2^{\frac{q}{2}-1}
n^{* \frac{q}{2}} \Gamma(\frac{q}{2})} \times C_{\alpha}
\label{eq:alphan}
\end{equation}

\noindent where $\alpha$ is the ``nominal"  type I error level, $q$ is the difference in dimension between the entertained models and $n^*$ is the \textit{Effective Sample Size} \citep{BBP2014}. This formula has (approximate) antecedents in \citet{CH1974} and \citet{Good1992}, and gives guidance on how to decrease the scale of p-values with the sample size. For one dimensional problems, it boils down to an elegant $\sqrt{ n^* \log(n^*)}$ correction formula. It can be used for finding the adaptive $\alpha_{n^*}$ levels and also the corresponding adaptive quantiles, suitable for testing via adaptive intervals for any q. (Incidentally note that the following formula gives a numerically more precise evaluation of the approximation (1) for $q=n_0=1$
\[
\alpha_e (n)=1-F_{\chi^2_1}(\chi_{\alpha}^2(1)+log(n))
\]
\noindent where $F_{\chi^2_1}$ is the distribution function of the Chi-square distribution with one degree of freedom)\\

Still the constant $C_{\alpha}$ has to be determined. In \citet{PP2014} propose using a \textbf{\textit{reference experiment}}, which is the result of a experimental design where the experimenter has specified a Type I error $\alpha$ and a Type II error $\beta$ for a specific point of statistical importance.

Let $n_0$ be the size of that reference experiment. Then

%
%
%
%
%
%

\[
\alpha(n)=\frac{\alpha*\sqrt{n_0\times (\log(n_0)+\chi^2_{\alpha}(1))} }{\sqrt{n^*\times (\log(n^*)+\chi^2_{\alpha}(1))}}.
\]

Table \ref{tab:adapalpha} shows how the adaptive alpha level changes with $n$ assuming $n_0=10$ and $\alpha=10$. Note that for $n=500$, $\alpha(n^*)$ has already drop to $0.005$.

\begin{table}[tbhp]
\begin{center}
\begin{tabular}{rl}

  \hline
  Sample Size &  $\alpha(n^*)$  \\
  \hline
  10 & 0.0500 \\
  50 & 0.0199 \\
  100 & 0.0135 \\
  500 & 0.0055 \\
  1000 & 0.0038 \\
  10000&0.0011\\
  \hline
\end{tabular}
  \end{center}
  \caption{Values of the adaptive alpha level $\alpha(n^*)$ for $n_0=10$ and $\alpha=0.05$.}
  \label{tab:adapalpha}
  \end{table}

Another possible route is adjusting $p$-values to obtain approximate posterior probabilities for the null hypothesis. In this line of thought, \citet*{SBB2001} introduced an easy to calculate link between p-values and Bayes factors (and posterior probabilities of hypothesis). We will call it the ``Robust Lower Bound" (RLB), as it provides a lower bound for the Bayes  which is valid in many situations.

\begin{equation} 
B_{01}  \geq  B_L(p_{val}) = -e  p_{val} \log_e(p_{val})
\end{equation}

From this expression, a lower bound on the posterior probability of the null hypothesis can be obtained.

\begin{equation}
 P(H_0|Data)\geq [1+\frac{1}{-exp(1)*Pval*log_e(Pval)}]^{{-1}}=\min{P(H_0|P_{val})}
\end{equation}
These bounds are very simple and can be easily calculated, but become less informative when $n$ increases.

Table \ref{tab:RLB}  shows the behavior of the RLB for different p-values. It introduces an important correction, broadly similar to the suggestions in \citet{Johnson2013}, but still static and inconsistent.


\begin{table}
\begin{center}
\begin{tabular}{lr}
$P_{val}$ &  $\min{P(H_0|P_{val})}$  \\
  \hline
   0.05&0.289  \\
   0.01&0.111  \\
   0.005& 0.067 \\
   0.001&0.018  \\
   0.0005& 0.010 \\
  0.0001&0.002\\
  \hline
\end{tabular}
\caption{Robust Lower Bounds for the null hypothesis (\citet{SBB2001}) for different p-values.}
\label{tab:RLB}
  \end{center}

  \end{table}

In the sequel, we propose a calibration of the Robust Lower Bound with the sample size which is simple, widely available and exhibits a reasonable behavior even for very small samples. We show that this calibration is consistent for large samples and as $p$-value$\rightarrow 0$.

\section{The Adaptive Robust Lower Bound ARLB}

In this section, we use the Robust Lower Bound of \citet{SBB2001} and the Adaptive Alpha Level in  \citet{PP2014} for obtaining a lower bound calibrated by the sample size $n$.
%
%
%

Consider the approximation of a log Bayes Factor for a null hypothesis vs an alternative hypothesis with difference in dimension equal to $q$ (i.e. the alternative has q more adjustable parameters):

\begin{eqnarray}
-2log(B_{01})&=&-2 \log(\frac{f_0(\mathbf{x}|\hat{\theta}_0)}{f_1(\mathbf{x}|{\hat\theta}_1)})-q\log(n^*) + C^* \nonumber \\
& \approx & \chi^2_{\alpha}(q)-q\log(n^*) + C^*. 
\end{eqnarray}

From this expansion $\alpha_{n^*}(q)$ as a function of the (effective) sample size $n^*$ is obtained and shown in equation (\ref{eq:alphan}) , so that if the $\alpha$ level is decreased as in (\ref{eq:alphan}) then the Bayes Factor is approximately constant as the information accumulates. 

Using these ideas, a calibration of the Robust Lower Bound could be obtained \noindent correcting the RLB (as a function of the sample size) so  that it converges to a constant when evaluated in  (\ref{eq:alphan}). The natural constant to choose is one, placing the hypotheses ``on equal footing". 

In symbols:

\vspace{0.25cm}
\begin{equation}
B_L(\alpha_{n^*}(q)) \times g_q(n^*)\rightarrow 1.
\label{eq:RLBcalib}
\end{equation}
\noindent where $B_L(\alpha) = -e \alpha \log \alpha$.

Clearly the stabilizing function $g$ is bound to depend (at least) on $n^*$ and $q$.\\
\begin{theorem}
The stabilizing function in (\ref{eq:RLBcalib}) is found to be:
\[
g_q(n^*)= \left[\frac{2 \cdot n^*}{\chi^2_{\alpha}(q)+q\cdot log(n^*)}\right]^{q/2} \cdot \frac{\Gamma(q/2)}{e}.
\]
\end{theorem}

\begin{proof} See Appendix A. \end{proof}



\vspace{0.25cm}

We define the adaptive odds  bound and probability, respectively, as
\[
O_L(\alpha,q,n^*)=B_L(\alpha) \cdot g_q(n^*)
\]
\[
P_L(\alpha,q,n^*)=\frac{1}{1+\frac{1}{O_L(\alpha,q,n^*)}}.
\]
$O_L$ and $P_L$ will be called the {\em Adaptive Robust Lower Bounds (ARLB) for the posterior odds and probabilities} respectively.

Note that, by construction, $\alpha_{n^*}(q)$ and $O_L(\alpha,q,n^*)$ should be broadly consistent. One advantage of the latter is that the robust lower bound automatically fixes the odds as equal to the bound for very small values of the sample size. \\
The result of Theorem 1 is somewhat surprising, and to the best of our knowledge novel. The one-dimensional case $q=1$, is illuminating,
\vspace{0.25cm}
\begin{equation}
O_L(\alpha,1,n^*)=[-e\cdot \alpha\cdot log(\alpha)]\times
\sqrt{\frac{2\cdot \pi\cdot n^*}{e^2\cdot (\chi^2_{\alpha}(1)+log(n^*))}}.
\label{eq:RLBcalibOneDimension}
\end{equation}
Notice that the correction of the Robust Lower Bound is of order of square root of $n^*$ over $logn^*$.
%
%
%
%

\subsection{Examples}

\subsubsection{Testing the value of a normal mean, variance known}

Let $X_1, X_2, \ldots, X_n$ be an i.i.d sample from $N(\theta,\sigma^2)$,  where $\sigma^2$ known. It is desired to test $H_0: \theta=\theta_0$ vs $H_1: \theta \neq \theta_0$.

Assume a prior for $\theta$ that is $N(\theta_0,\tau^2)$. It is also usual to select $\tau^2 = k\sigma^2$. We will use $k=1$ and  $k=2$. This last value corresponds to  the intrinsic prior for $\theta$ \citep{BerPer14}. 

The Bayes factors obtained using these priors have the form

\[ B_{01} = \sqrt{1+kn}\exp\left(\frac{-\frac{1}{2}z^2}{1+1/(kn)} \right)\]

\noindent where $z=\sqrt{n}(\bar{x}-\theta_0)/\sigma$

We  will compare the Bayes Factors for these priors with the RLB and the ARLB.

Figure \ref{fig:normal} shows the posterior probability for the null hypothesis $H_0$ for $n=50$ and $n=500$ using both priors, the Robuts Lower Bound and the ARLB. Note that $P_L(\alpha,q,n^*)$ looks very similar to the result obtained using a normal prior with the same variance as the likelihood ($k=1$), and leads to an inference very close, though slightly less conservative, than the intrinsic prior.

Note that for $n=50$ a $p$-value of $0.05$  corresponds to a posterior probability for $H_0$ around $0.5$. This is clearly not enough for rejecting the null hypothesis. 

\begin{figure}[tbhp]
\begin{center}
\begin{tabular}{cc}
\includegraphics[width=0.42\linewidth]{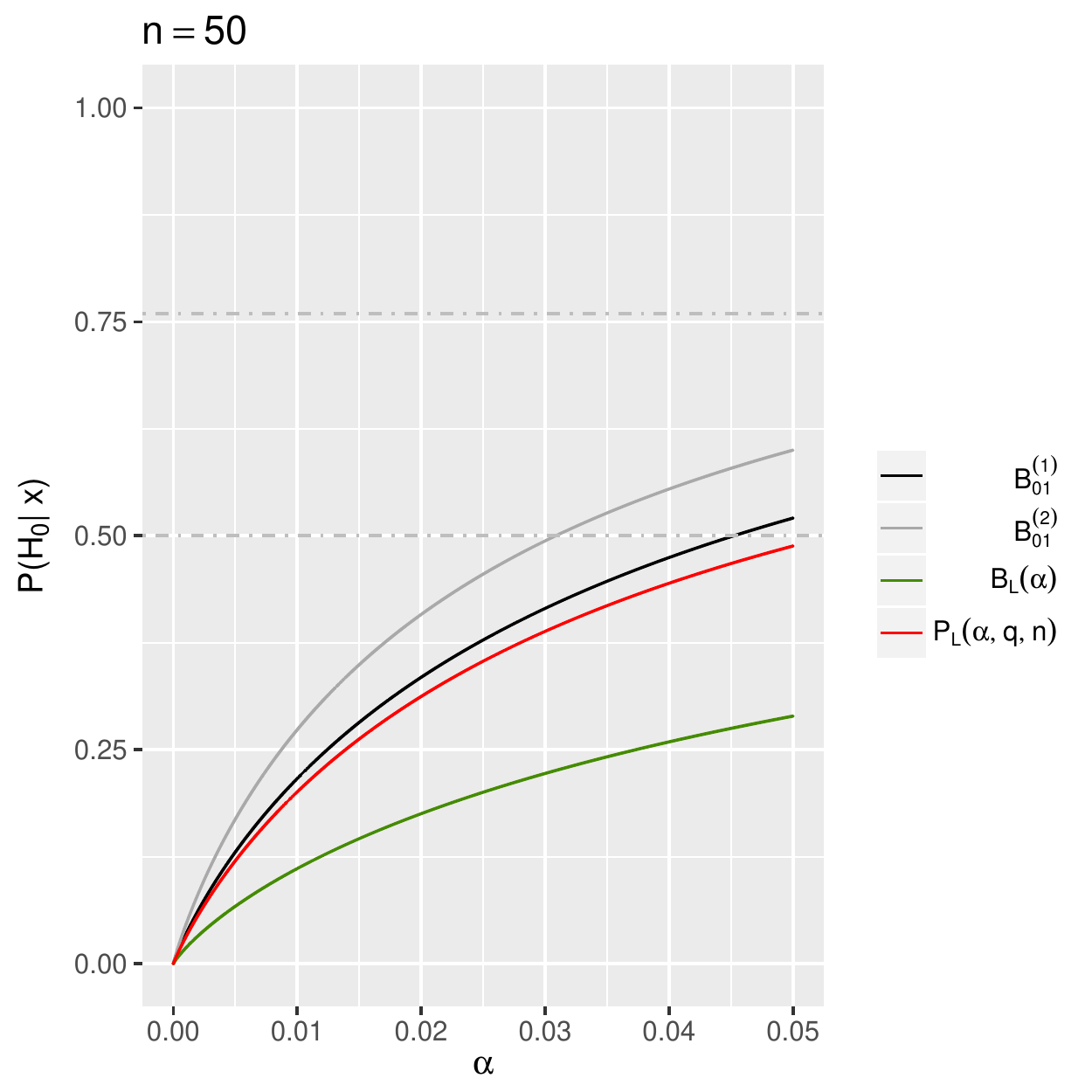} & \includegraphics[width=0.42\linewidth]{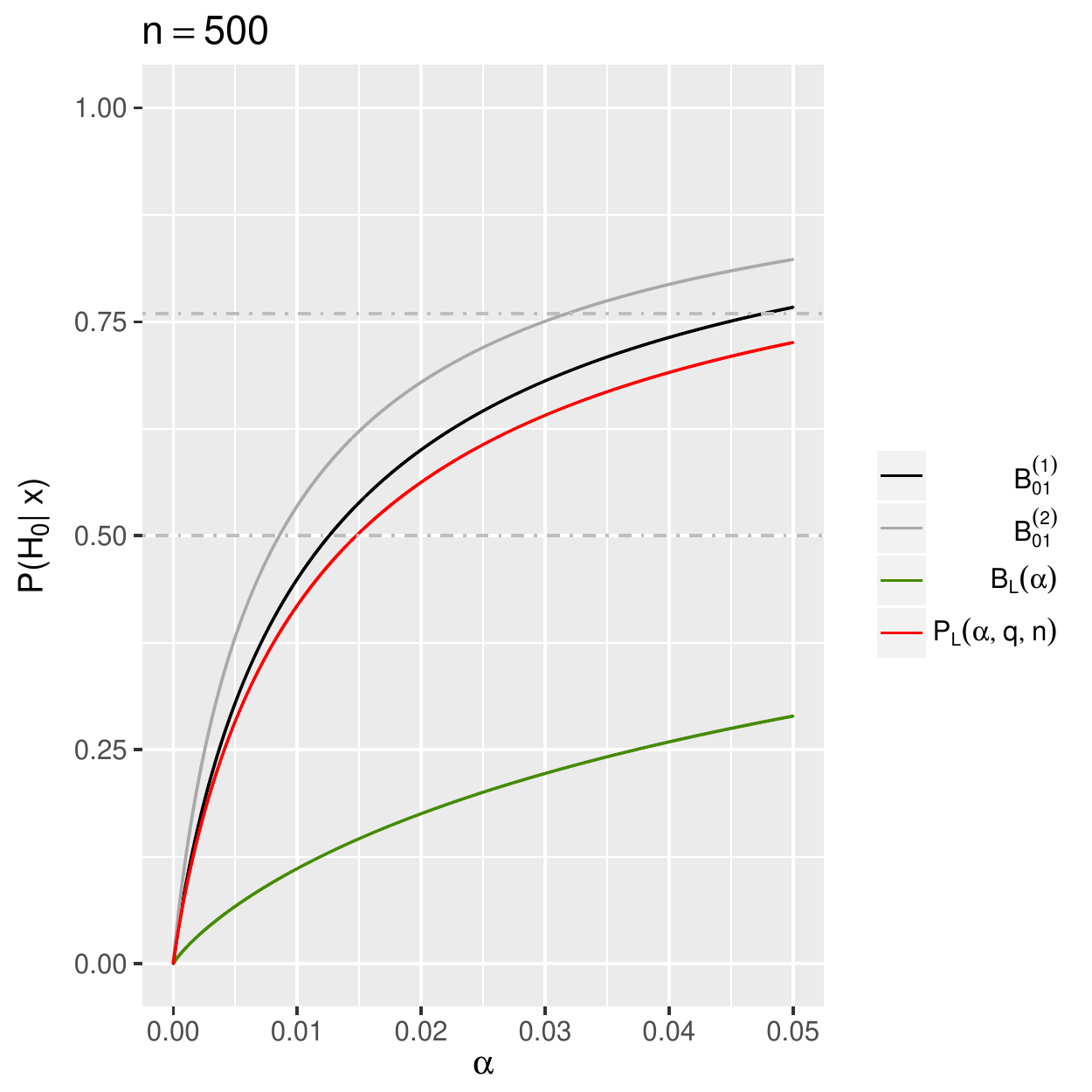}\\
\end{tabular}
\end{center}
\caption{Posterior probability for the null hypothesis $H_0$ for $n=50$ and $n=500$ using normal priors, the Robust Lower Bound and the ARLB ($P_L(\alpha,q,n^*)$) }
\label{fig:normal}
\end{figure}

\subsubsection{Testing the value of a normal mean, variance unknown}
Let $X_1, X_2, \ldots, X_n$ be an i.i.d sample from $N(\theta,\sigma^2)$, $\sigma^2$ unknown. As before, It is desired to test $H_0: \theta=\theta_0$ vs $H_1: \theta \neq \theta_0$.

\citet{Almodovar} presents an approximation for the Bayes factor based on the intrinsic prior

\begin{footnotesize}
\[ B_{01}^{IP} \approx \sqrt{n}\left(1+\frac{t^2_{n-1}}{n-1} \right)^{-n/2} \cdot \frac{t^2_{n-1}}{n-1} \cdot \frac{1}{1+\exp(1+t^2_{n-1}/(n-1))} \]
\end{footnotesize}

\noindent where  $t_{n-1} = \sqrt{n} \cdot \frac{(\bar{x}-\theta_0)}{s}$ is the usual $t$ statistic.

Another objective Bayes Factor is based on the Robust prior proposed by \citet{Berger1985}
\begin{footnotesize}
\[
B^R_{01}= \sqrt{\frac{2}{n + 1}} \left( \frac{n-2}{n-1} \right) (t_{n-1})^2 \left(1+\frac{(t_{n-1})^2}{n-1} \right)^{-\frac{n}{2}} \left(1- \left(1+2 \frac{(t_{n-1})^2}{n^2-1} \right)^{-\frac{n-2}{2}} \right)^{-1}
\]
\end{footnotesize}

When comparing these objective Bayes factors with the ARLB, we can see that the behavior of the later is very close to the approximated intrinsic Bayes factor in \citet{Almodovar} and compatible with the Bayes factor based on Berger's robust prior. 

\begin{figure}
\begin{center}
\begin{tabular}{cc}
\includegraphics[width=0.38\linewidth]{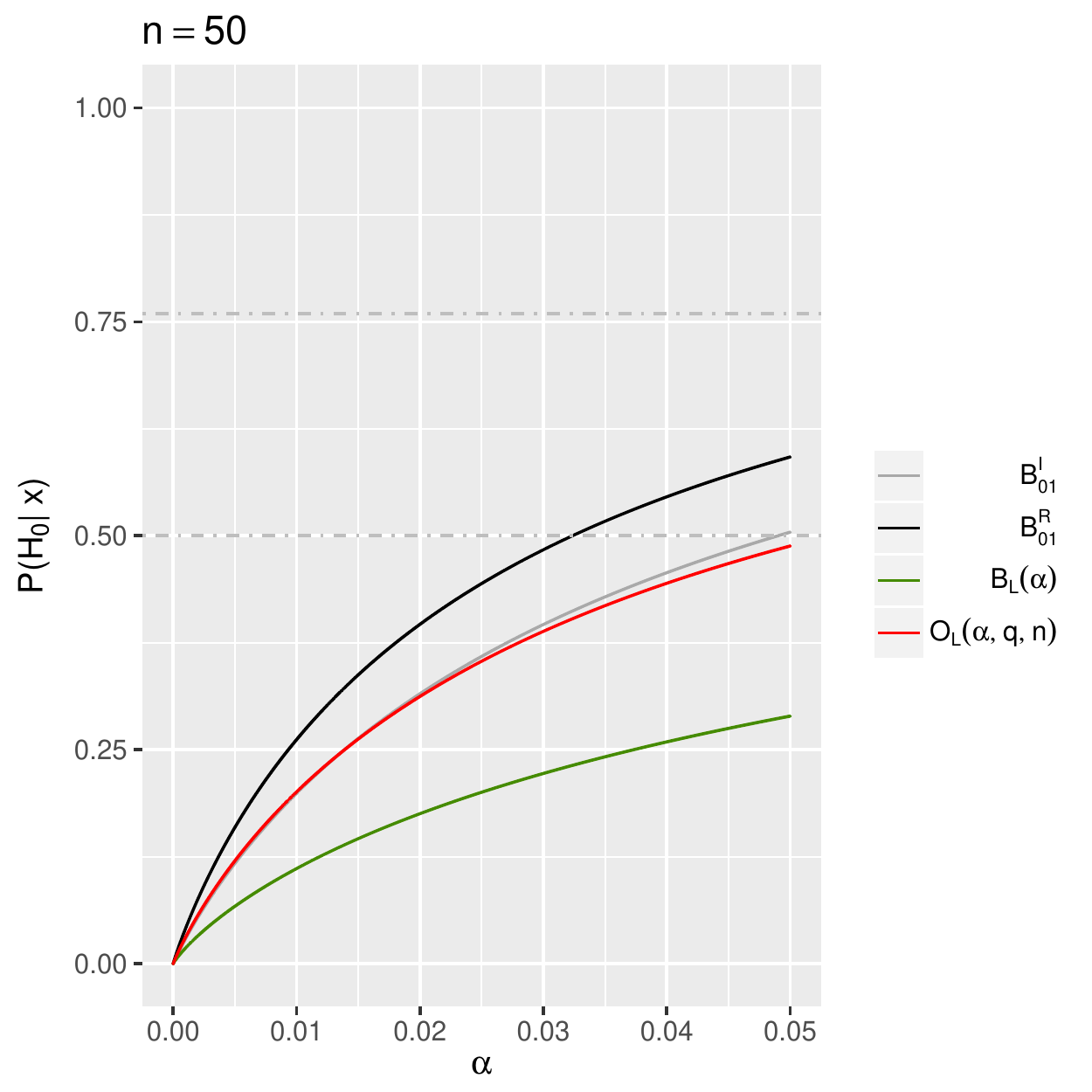} & \includegraphics[width=0.38\linewidth]{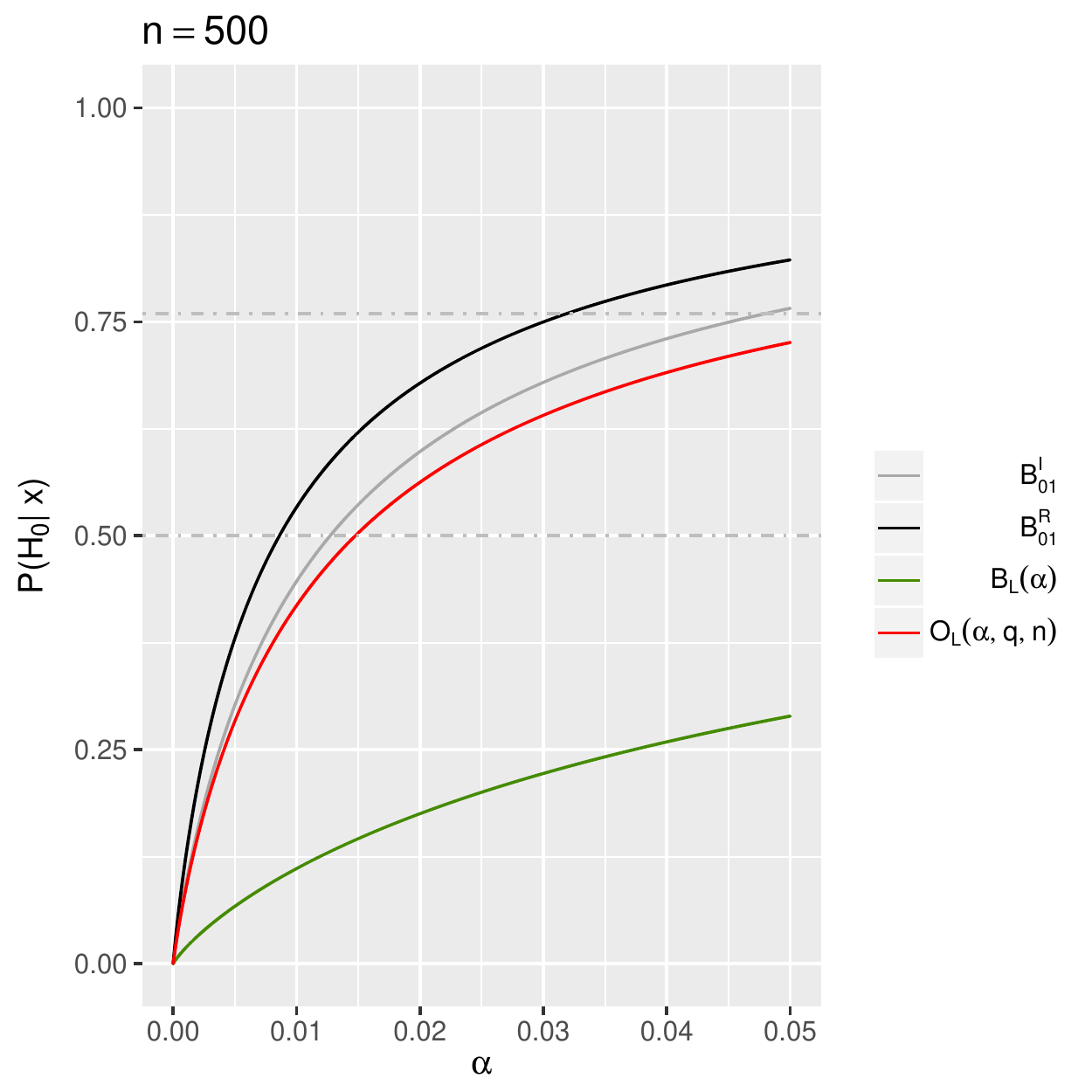}\\
\end{tabular}
\end{center}
\caption{Posterior probability for the null hypothesis $H_0$ for $n=50$ and $n=500$ using Bayes factors corresponding to the Intrinsic prior \citep[approximation in][]{Almodovar}, Berger's Robust prior \citep{Berger1985}, the Robust Lower Bound and the ARLB ($P_L(\alpha,q,n^*)$) }
\label{fig:normalsigma}
\end{figure}

\subsubsection{A non normal example: exponential distribution}
Let $X_1, X_2, \ldots, X_n$ be a random sample from an Exponential($\lambda$) distribution, with $f(x|\lambda) = \lambda e^{-\lambda x}$. Assume we want to test the hypotheses

\[ H_0: \lambda=\lambda_0 \mbox{   vs   } H_1: \lambda \neq \lambda_0\] 

In \citet{PericchiA} the intrinsic prior for this comparison is shown to be
\[ \pi^I(\lambda) = \frac{\lambda_0}{(\lambda+\lambda_0)^2} \]

\noindent This prior is a particular case of the Scaled Beta 2 studied by \citet{PPRSB2}. In particular, the intrinsic prior for testing the parameter of an exponential distribution is a $Sbeta2(1,1,\frac{1}{\lambda_0})$.

The Bayes factor is calculated as

\[ B^I_{01} = \frac{m_0(\mathbf{x})}{m_1(\mathbf{x})}\]

\noindent where
\begin{eqnarray*}
m_0(\mathbf{x})&=&\lambda_0^n \exp(-n \lambda_0 \bar{x})\\
m_1(\mathbf{x}) &= & \frac{ \Gamma(n)}{ \lambda_0(n\bar{x})^{n-1}} \left\{n(\frac{\bar{x}}{\lambda_0}+1)\mathrm{e}^{n\frac{\bar{x}}{\lambda_0}}E_{n}(n\frac{\bar{x}}{\lambda_0})-1\right\}\\
\end{eqnarray*}

\noindent and $E_n(x)=\int_1^{\infty} \frac{\exp(-xt)}{t^n}dt$ is the {\em Exponential Integral Function} \citep[][section 5.1]{Abra} . For this last calculation, see Appendix B. 

For the graphs shown in Figure \ref{fig:exp}, we used $\lambda_0=1$, and the value of $\bar{x}$ was approximated using the chi-square approximation to the likelihood ratio test. We display results obtained using the lower solution, but the graphs obtained with the upper solution are completely equivalent. It can be seen that the Adaptive Robust Lower Bound follows very closely the conclusions obtained with the Bayes factor based on the intrinsic prior.

\begin{figure}
\begin{center}
\begin{tabular}{cc}
\includegraphics[width=0.38\linewidth]{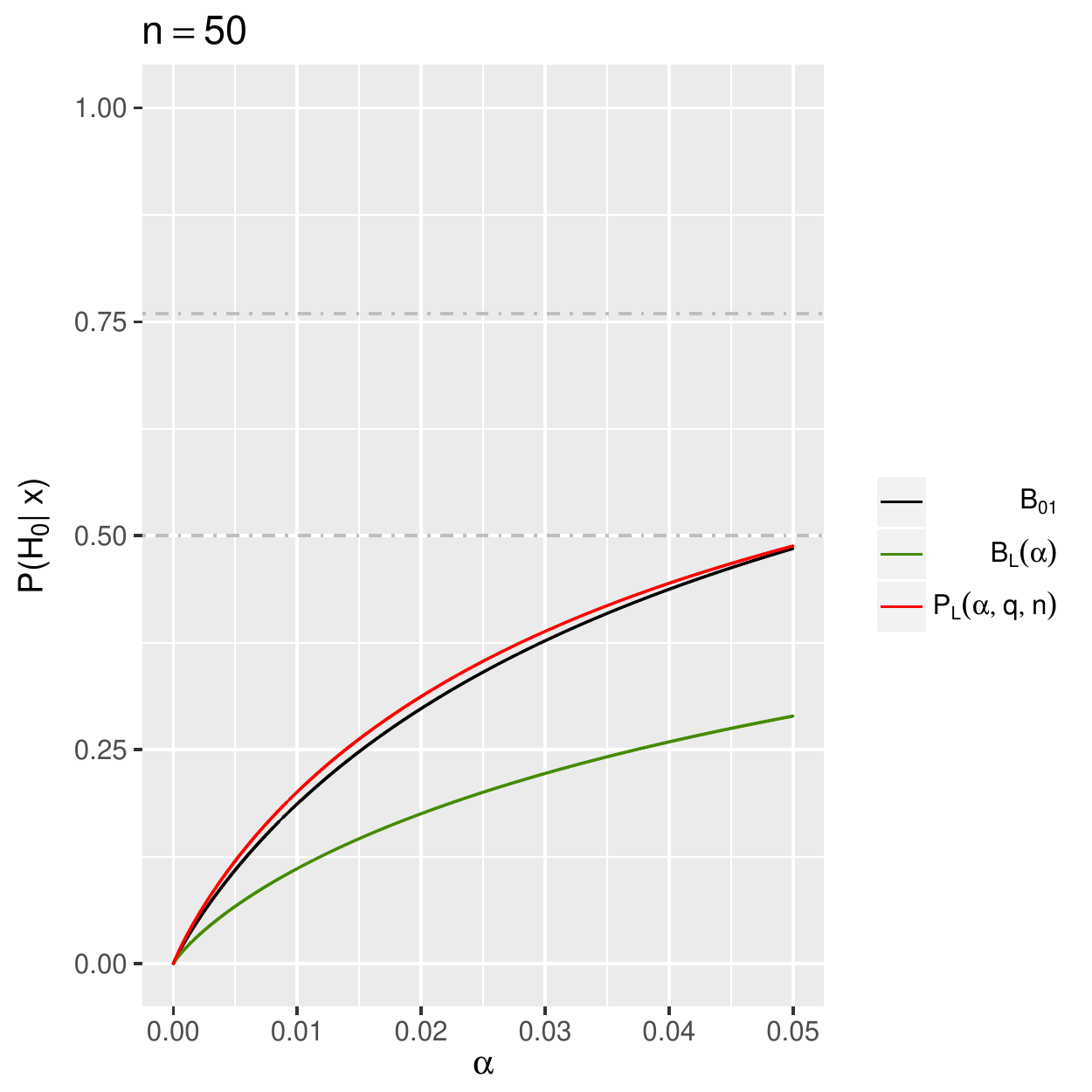} & \includegraphics[width=0.38\linewidth]{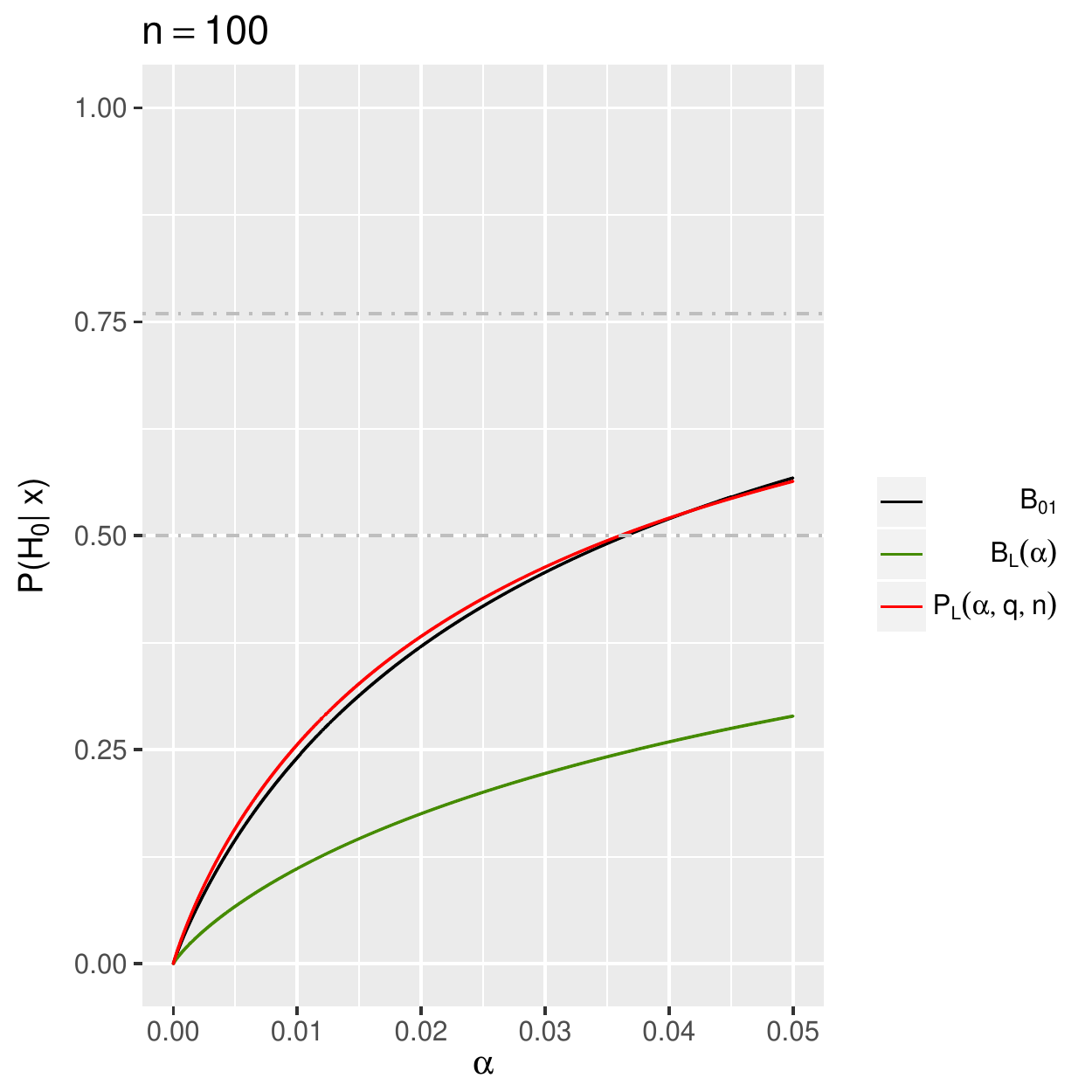}\\
\end{tabular}
\end{center}
\caption{Posterior probability for the null hypothesis $H_0: \lambda_0=1$ for $n=50$ and $n=100$ using the Bayes factors corresponding to the Intrinsic prior \citep{PericchiA},  the Robust Lower Bound and the ARLB ($P_L(\alpha,q,n^*)$) }
\label{fig:exp}
\end{figure}

\subsubsection{Models differing in more than one parameter: comparing nested linear models}

In \citet{Giron}, the authors approach an objective Bayesian setting for the problem of comparing nested linear models.

\begin{eqnarray*}
\bm{M_0}: \bm{N}_n(\bm{y}| \bm{X}_1\bm{\gamma_1},\sigma_0^2 \bf{I}_n) &,&\pi_0^D(\bm{\gamma_1},\sigma_0) \propto \frac{c_0}{\sigma_0^{q_0}}\\
\bm{M_0}: \bf{N}_n(\bm{y}| \bm{X}\bm{\alpha},\sigma_1^2 \bf{I}_n) &,&\pi_0^D(\bm{\alpha},\sigma_1) \propto \frac{c_1}{\sigma_0^{q_1}}\\
\end{eqnarray*}

\noindent where $\bm{X}$ is a $n\times k$ design matrix and  $\bm{X}=(\bm{X}_0|\bf{X}_1)$, with  $\bm{X}_1$ is a $n\times k_1$ submatrix of $\bm{X}$. $q_0$ and $q_1$ are non-negative integers
(usual reference priors correspond to $q_0 = q_1 = 1$, Jeffreys priors correspond to $q_0 = k_1 + 1$ and $q_1 = k + 1$,  and modified Jeffrey priors correspond to $q_0=1$ and $q_1=k-k_1+1$).

The authors find intrinsic priors for this setting and under certain conditions they calculate the following expression for the intrinsic Bayes factor

{\footnotesize
\[
B_{10}(\bm{D}) = \frac{2(k+1)^{k_0/2}}{B(q_1/2,1/2)} \int_0^{\infty} \frac{(\sin \phi)^{k_0+q_0-1} }{(\cos \phi)^{k_1+q_1-1}} \frac{(n+(k+1)\sin^2\phi)^{(n+q_0-k-1)/2}}{(n{\cal B}_n +(k+1)\sin^2\phi)^{(n+q_0-k_1-1)/2}} d\phi
\]
}

Here ${\cal B}_n=\frac{SC}{SC_1}$ with {\footnotesize $SC= \bm{y}' (I+\bm{X}(\bm{X}'\bm{X})^{-1}\bm{X}')\bm{y}$ } ,{\footnotesize $SC= \bm{y}' (I+\bm{X}_1(\bm{X}_1'\bm{X}_1)^{-1}\bm{X}_1')\bm{y}$ }

The authors provide calibration curves for $p$ values based on this result.

In Figure \ref{fig:lm} we compare these intrinsic Bayes factors based on reference priors, Jeffreys priors and modified Jeffreys priors with the calibration of the lower bound and with $BIC(\alpha)$. Note that for $q=1$ both BIC and ARLB have a similar behavior, which leads to conclusions similar to those obtained with exact Bayes factors. When $q=3$, though, different priors lead to different conclusions, with the ARLB providing inference similar to the Intrinsic Prior based on the Reference prior and BIC being more compatible with the inference obtained with the Jeffreys and the Modified Jeffreys prios.

\begin{figure}
\begin{center}
\begin{tabular}{cc}
\includegraphics[width=0.38\linewidth]{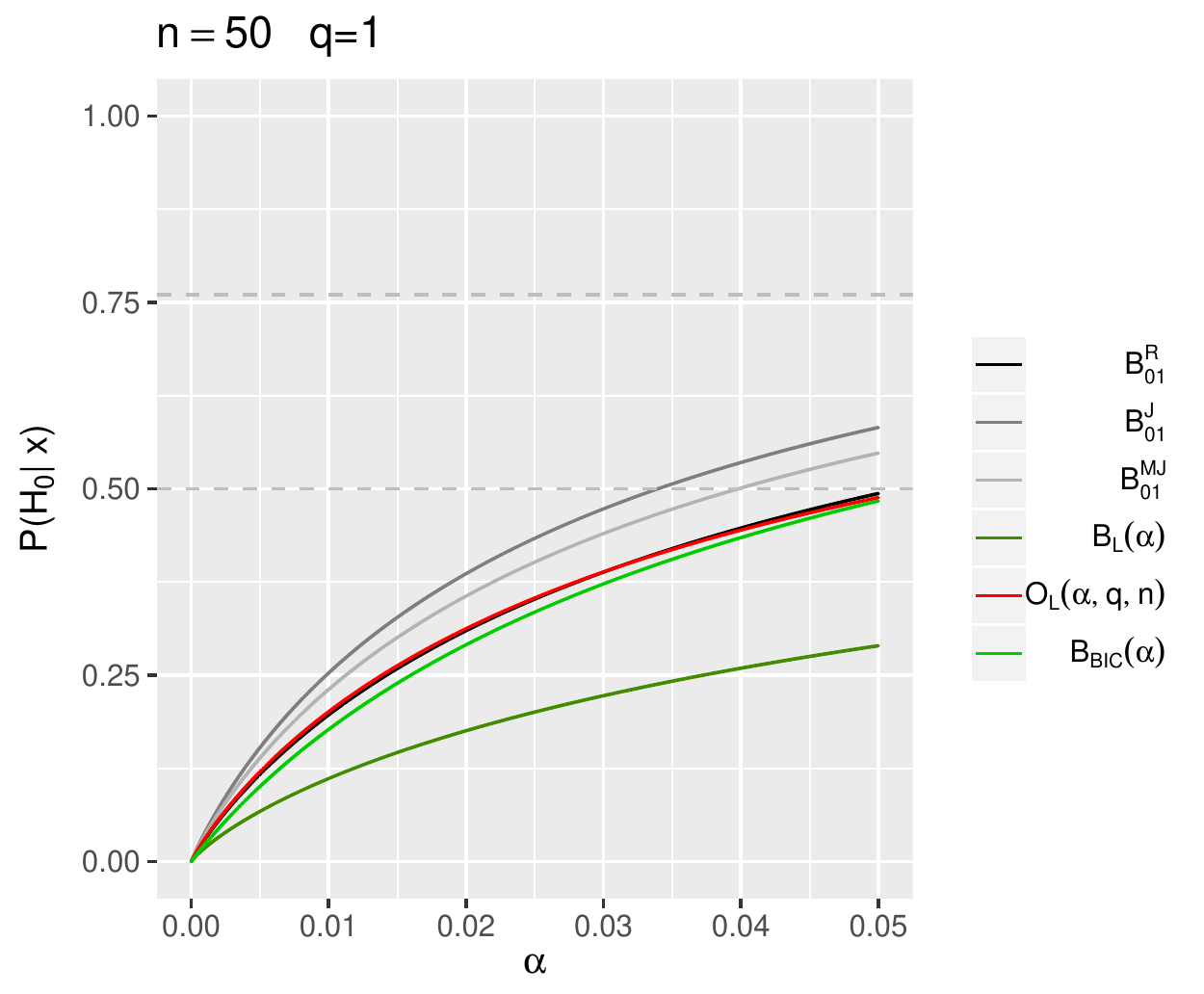} & \includegraphics[width=0.38\linewidth]{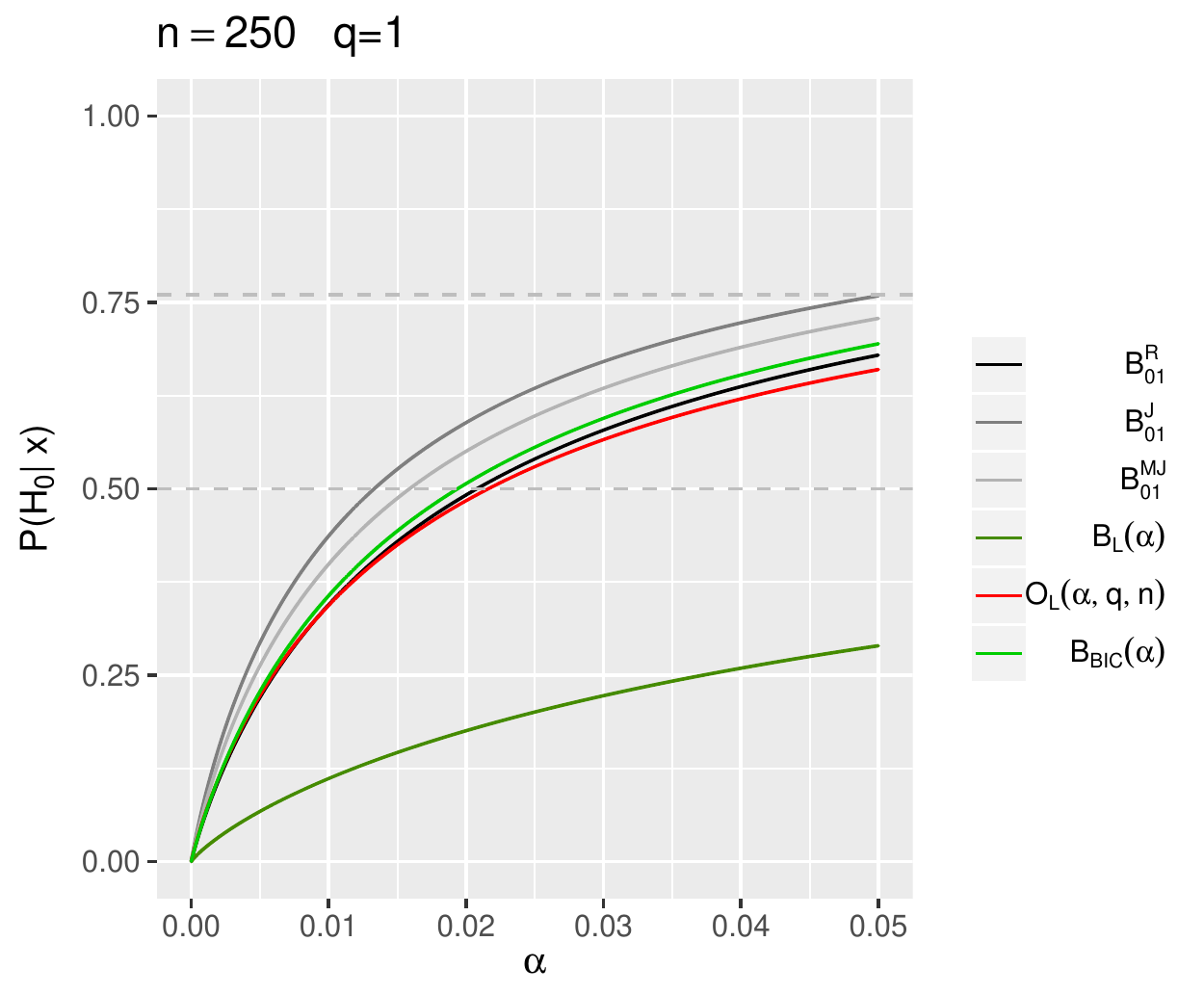}\\
\includegraphics[width=0.38\linewidth]{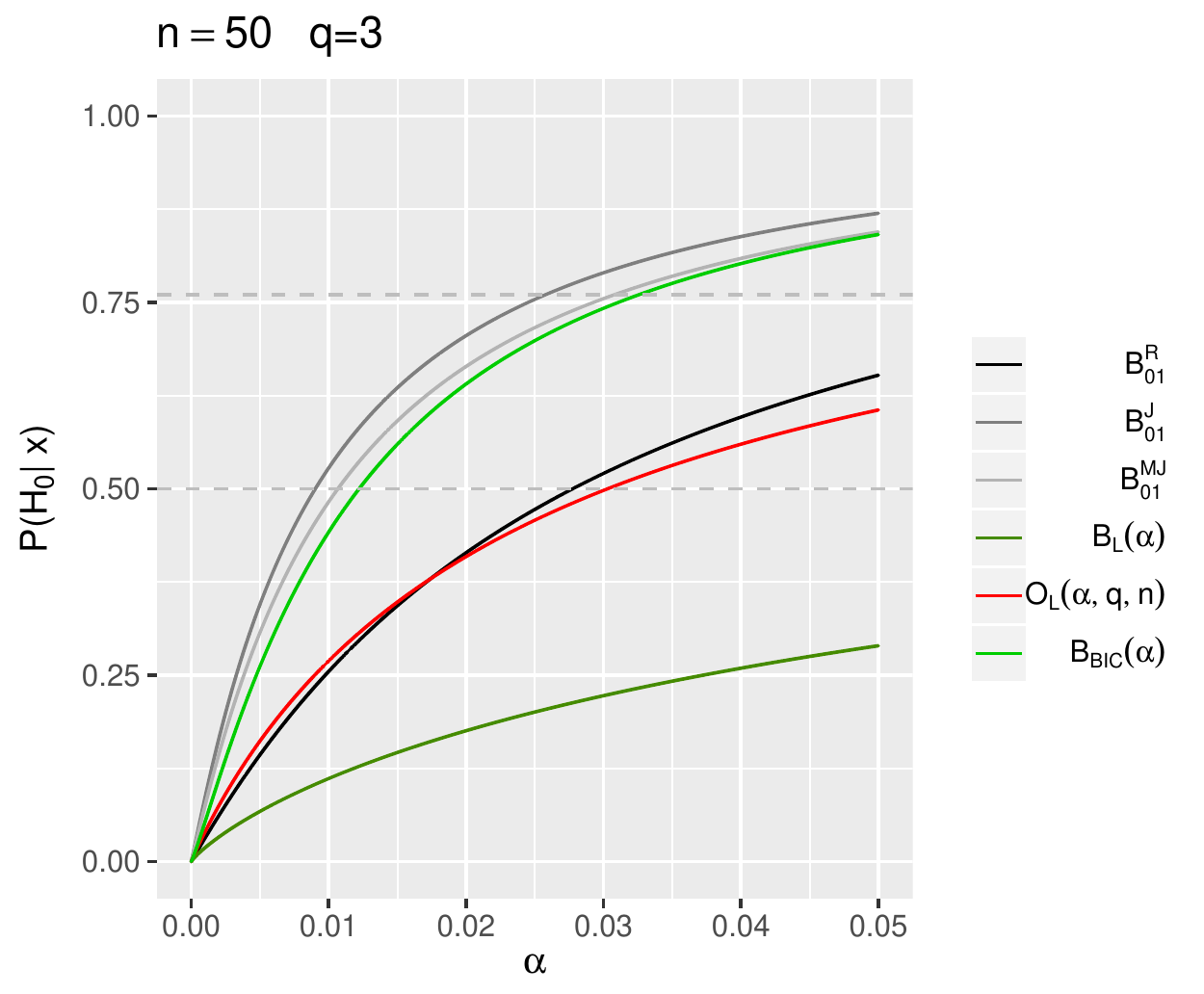} & \includegraphics[width=0.38\linewidth]{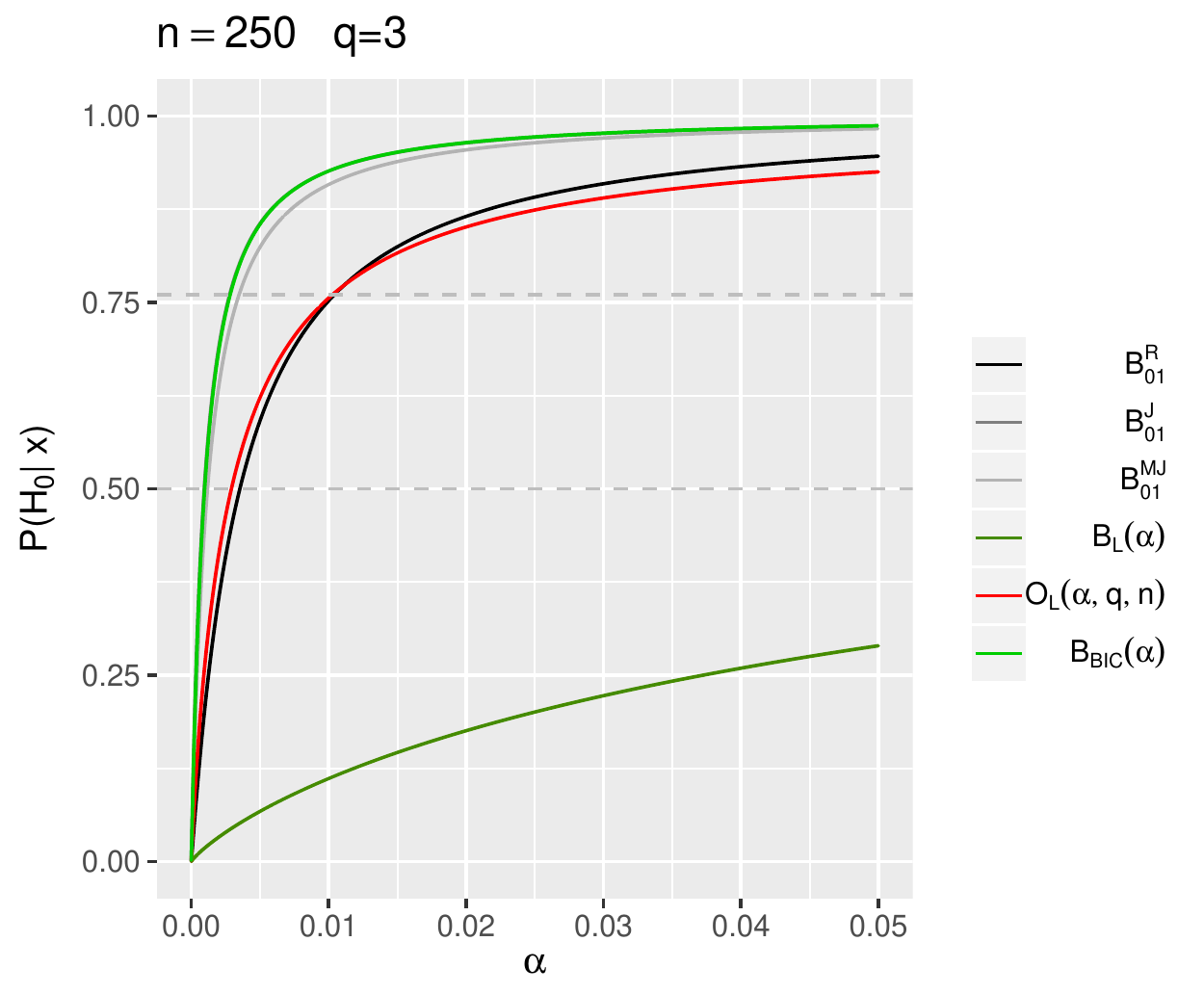}\\
\end{tabular}
\end{center}
\caption{Posterior probability of  $H_0$ using intrinsic Bayes factors based on Reference priors, Jeffreys priors and modified Jeffreys priors \citep{Giron}, BIC, the Robust Lower Bound and the ARLB ($P_L(\alpha,q,n^*)$) .}
\label{fig:lm}
\end{figure}

\subsubsection{An application: Hald's data}

Consider Hald's regression data (Draper and Smith, 1981), consisting of $n = 13$ observations on a dependent variable $y$, with four potential regressors: $x_1$, $x_2$, $x_3$, $x_4$.

Here we will follow an ``encompassing approach" (as Berger and Pericchi 1996), comparing all models with the full model
\[y=c+\beta_1 x_1 + \beta_2 x_2 + \beta_3 x_3 +\beta_4 x_4 + \varepsilon, \varepsilon \sim N(0,\sigma^2) \]

We will consider only those cases where the p-value is lower that $\frac{1}{e}$ (so the robust lower bound is valid)

{\tiny
\begin{center}
\begin{tabular}{lrrrrrrr}
\hline
&\multicolumn{3}{c}{IBF}&& & 	&	\\
\cline{2-4}
Model &	\multicolumn{1}{c}{Ref.} & \multicolumn{1}{c}{Jeffreys} & \multicolumn{1}{c}{M. Jeff.} &\multicolumn{1}{c}{$B_{lb}(p)$}& \multicolumn{1}{c}{$O_L(p,q,n)$}  & \multicolumn{1}{c}{$B_{(BIC)}$} & \multicolumn{1}{c}{p-val}\\
\hline
234c & 0.60776 &0.73697&0.67515&{\bf 0.50970}&0.70192&{\color{red} 0.21582}&0.07082\\
13c  & 0.00018 & 0.00009 & 0.00010 & {\bf 0.00008} &0.00008 & {\color{red} 0.00000} & 0.00000\\
14c	& 1.15714 & 2.31974 & 1.76313 & {\bf 0.81460} & 0.89583  & {\color{red} 0.71623} & 0.16800\\
23c	& 0.00586 & 0.00335 & 0.00386 & {\bf 0.00414} & 0.00414  & {\color{red} 0.00001} & 0.00018\\
24c	& 0.00056& 0.00030 & 0.00031 & {\bf 0.00029} & 0.00029	  & {\color{red} 0.00000} & 0.00001\\
34c	& 0.07650 & 0.06740 & 0.07492 & {\bf 0.07782} & 0.07782  & {\color{red} 0.00277} & 0.00550\\
1c	& 0.00030 & 0.00014 & 0.00014 & {\bf 0.00016} & 0.00016 & {\color{red} 0.00000} & 0.00000\\
2c	& 0.00089 & 0.00043 & 0.00046 & {\bf 0.00055} & 0.00055 & {\color{red} 0.00000} & 0.00002\\
3c	& 0.00007 & 0.00003 & 0.00003 & {\bf 0.00003} & 0.00003  & {\color{red} 0.00000} & 0.00000\\
4c	& 0.00096 & 0.00047 & 0.00050 & {\bf 0.00061} & 0.00061  & {\color{red} 0.00000} & 0.00002\\
\hline
\end{tabular}
\end{center}}

\begin{itemize}
\item The values of the lower bound and its adjustments are compatible with the intrinsic Bayes factors.
\item Values of $B_{BIC}$ are uniformly lower than the Robust Lower Bound of Sellke et al. (sample size $n=13$).
\end{itemize}

\section{Consistency and Robustness}
The Adaptive Robust Lower Bound has two important {\em consistency properties} that yield weight to its justification as a safe measure of evidence of a Hypothesis. Even some {\it{exact}} Bayes Factors (i.e. Bayes Factors that arise with proper, usually conjugate, priors)
does not enjoy these two desirable consistency properties,
a case in point are the popular "g-priors", which fails the second consistency test, Berger and Pericchi (2001). Here we make the simplifying and usual assumption, that the sampling model does not change as the sample size grows. We also assume that the sampling model is a candidate model.

\begin{definition}{ \underline{Large Sample Consistency}:}
A measure of {\it{Probability Evidence}} is {\it{large sample consistent}} if as the sample size grows the probability of the sampling model converges to one, and the probability of all other alternative models converges to zero.
\end{definition}
\begin{definition}{\underline{Large Likelihood Consistency}:}
A measure of  {\it{Probability Evidence}} is {\it{large likelihood consistent}} for the probability evidence goes to one when the likelihood ratio of
it against any competitor model goes to infinity.
\end{definition}
\begin{remark}
Real Bayes Factors, for a fixed number of parameters as the sample grows, obey large sample consistency, but not necessarily large likelihood
consistency, Berger and Pericchi (2001). It is therefore an splendid recommendation for the Adaptive Lower Bound.
\end{remark}
\begin{remark}
BIC, with the correct notion of "Effective Sample Size",  Berger, Bayarri and Pericchi (2013)  obeys both kinds of consistency, but can be a very bad approximation for small to moderate data sets. On the other hand the Robust Lower Bound is valid and informative, for virtually all sample sizes. The Adaptive Robust Lower Bound takes advantage of the Roibust Lower Bound for small samples and of the Bayes Factor Asymptotics for moderate to large sample sizes.
\end{remark}

Before proceeding it is convenient to prove a useful lemma about the Robust Lower Bound.

\begin{lemma}
$B_L(p_{val})=-e \cdot  p_{val} \cdot log(p_{val}) \ge e \cdot  p_{val} >p_{val}$, $\mbox{   }$ for $\mbox{   }$,  $0<p_{val}<1/e$, where $e$ is Euler constant $e=\exp(1)$.
\end{lemma}
\begin{proof}Lets call $g(p_{val})=B_L(p_{val})/p_{val}=-e \log(p_{val})$. Note that $g(p_{val})$ is equal to $e$ for $p_{val}=1/e$ and taking the limit as $p_{val}\rightarrow 0$, $g(p_{val})$ grows without bound. Now taking the derivative of   $g(p_{val})$ with respect to $p_{val}$ yields $-e / p_{val}$, which is negative, and thus $g$ is decreasing with minimum at $1/e$. Therefore, $g(p_{val})  \ge e$ which implies $B_L(p_{val}) \ge e \cdot p_{val}>  p_{val}$.
\end{proof}
\vspace{0.25cm}

Before proceeding we state another useful Lemma, the inequality of quantile and p-values. We set the discussion for one sided alternatives, that is $H_0: \theta=0 \mbox{  VS  } H_1:\theta > 0$. Adapting to $H_0: \theta=0 \mbox{  VS  } H_1: \theta \neq 0$, is simple.

\begin{lemma}
$Z_{p-{val} }< \sqrt{(-2 \cdot log(p_{val}))}$, $\mbox{   }$,  for $Z_{p-{val}}>1/\sqrt{2\cdot \pi}$, or equivalently $0<p_{val}<1-\Phi(1/\sqrt{2\cdot \pi})=0.345$.\\
\end{lemma}
\begin{proof} The starting point is the well known Mill's Ratio inequality,
\[
p_{val}=1-\Phi(Z_{p-{val}})\leq \frac{\phi(Z_{p-{val}})}{Z_{p-{val}}}=\frac{exp(-Z^2_{p-{val}}/2)}{\sqrt{2\cdot \pi}\cdot z_{p-val}}.
\]
After some algebra, the following inequality appears,
\[
\sqrt{-2\cdot log(p_{val})} \geq \sqrt{[Z^2_{p-val}]+[2\cdot log((\sqrt{2\cdot \pi}\cdot Z_{p-val}))]}.
\]
Now the second term in the right hand side is non negative (since the quantil is increasing with the decrease of $p_{val}$, and therefore for $Z_{p-val} \geq 1/\sqrt{2\cdot \pi}$),
 $Z_{p-val} < \sqrt{(-2 \cdot log(p_{val}))}$ establishing the result.

Recall that the adaptive bounds in the odds scale is \\  $O_L(p_{val})=B_L(p_{val})\cdot (\frac{n^*}{log(n^*)})^{q/2}$. In what follows we will suppose that $q=1$ and the sampling model is Normal with known variance, but the results are conjectured to hold much more generally.
\begin{theorem}The ARLB is large sample consistent, for $q=1$ and Normal Data.\end{theorem}
\begin{enumerate}
\item Assume that the sampling model is the Null Model. Recall that under the Null the  $p_{val}$   is uniformly distributed. We wish to prove that, for any $W>0$,
\[
P(O_L(p_{val})>W|H_0)) \rightarrow 1,
\]
equivalently
\[
P(O_L(p_{val}) \le W|H_0) \rightarrow 0,
\]
or using the Lemma 1, above,
\[
P(B_L(p_{val}) (\frac{n^*}{log(n^*)})^{q/2} \le W|H_0))\le P(p_{val} \le (\frac{W}{e})\cdot (\frac{log(n^*)}{n^*})^{q/2})= (\frac{W}{e}) \cdot(\frac{log(n^*)}{n^*})^{q/2},
\]
(since the $p_{val}$ is Uniform under the Null) which goes to zero as the sample size $n$ grows.
\item Assume that the sampling model is the Alternative Model.
We wish to prove that for any $W>0$
\[
P(O_L(p_{val} )< W|H_1) \rightarrow 1.
\]
Equivalently,
\[
P(B_L(p_{val})<W \cdot \sqrt{\frac{log(n^*)}{n^*}}|H_1),
\]
converges to zero. But this expression,  using Lemma 1 is less than
\[
P(P_{val} \le \frac{W}{e}\cdot \sqrt{(log(n^*)/n^*)}|H_1).
\]
Let us call $\delta=\mu/\sigma>0$ the ratio of the true mean over the standard deviation. Then the distribution of the $P_{val}$ is \citep{Hung1997},
\[
G_{\delta}(P_{val}|H_1)=1-\Phi(Z_{P_{val}}-\sqrt{n^*}\cdot \delta)=
\]
\[
1-\Phi(Z_{\frac{W}{e}\cdot \sqrt{log(n^*)/n^*}}-\sqrt{n^*}\cdot \delta).
\]
Thus all that  is needed is to prove that,
\[
\Phi(Z_{\frac{W}{e}\cdot \sqrt{log(n^*)/n^*}}-\sqrt{n^*}\cdot \delta) \rightarrow 0.
\]
Now we make use of Lemma 2, which ensures that the expression above is less than,
\[
\Phi(\sqrt{-2\cdot log(\frac{W}{e}\cdot \sqrt{(\frac{log(n^*)}{n^*})})}-\sqrt{n^*}\cdot \delta),
\]
and this expression goes to zero, since
\[
constant+\sqrt{log(n^*)}-\sqrt{n^*} \cdot \delta \rightarrow -\infty.
\]
\end{enumerate}
\end{proof}

\begin{theorem}
The ARLB  is large likelihood consistent.
\end{theorem}

\begin{proof}
For sample size n fixed, assume that $Lik.Ratio_{01}(x_{observed})\rightarrow 0$, that is
$-2 \cdot Lik.Ratio_{01}(x_{observed})\rightarrow \infty$. Now,
\[
p_{val}=Pr(-2log(LR_{0,1}(X))>-2 log(LR _{0,1}(x_{observed}))|H_0) \rightarrow 0,
\]
since $-2log(LR_{0,1}(X))$ is approximately distributed as a central $\chi_r^2$, under $H_0$. But, applying L'Hopital's rule, as $p \rightarrow 0$,
\[
lim _{p \rightarrow 0}B_L(p)=lim _{p \rightarrow 0}\frac{e \cdot log(p_{val})}{-1/p}  =lim _{p \rightarrow 0} e \cdot p=0,
\]
therefore (for fixed $n^*$)  $O_L(p)=B_L(p)*(\frac{n^*}{log(n^*)})^r \rightarrow 0.$
\end{proof}

\begin{remark}
The p-value is not well defined, when the Null Hypothesis
is not simple, that is when it is a composite Null Hypothesis. For discussion of a general definition of a valid $p_{val}$, see Bayarri and Berger (1998) and Robins et al (2000)In fact The most frequent definition of a general p-value is:
\[
p^*_{val}=Sup_{\theta \in \Theta_0}P(T(X)\geq T(x)|\theta ).
\]
The question is if $p^*_{val}$ is a valid $p_{val}$? The stringent definition is that a valid $p_{val}$ is an statistic uniformly distributed under $H_0$. The relaxed definition is that the distribution $p_{val}$ is smaller or equal to $p$. In other words the stringent definition  requires that the distribution is equal to $p_{val}$ and the relaxed definition is that is equal or smaller to $p_{val}$. The $p^*_{val}$ obeys the definition of a relaxed $p_{val}$,
for a proof see Casella and Berger (2002) page 397\\
The question is if the Robust Lower Bound is still valid for $p*_{val}$?
\end{remark}
Under the stringent definition of a $p_{val}$ the Robust Lower Bound is valid lower bound for "decreasing failure rates alternatives". The conjecture is that it is still valid (at least for not very small sample sizes), under the relaxed definition, in which case the results established here would be still valid, for the relaxed definition.

\section{Final Comments}

P-values may have been misunderstood, but they are available for virtually any statistical model. So, calibration of (conditional on the data) p-values, may ironically be the fastest way to popularize the use of Posterior Model Probabilities.

The results established in this paper make the Adaptive Robust Lower Bound an easy to implement and safe measure of evidence for or against a null hypothesis that automatically calibrates p-values as posterior probabilities, taking into account the effect of the sample size.

\section*{Acnowledgements}
This work has been partially funded by NIH grants  U54CA096297, P20 GM103475-14A1 and R25MD010399. The authors want to thank Prof. James O. Berger for many illuminating discussions on the subject. The contributions of Susie Bayarri have been an inspiration during all this proccess, and authors dedicate this work to her memory.

\bibliographystyle{ba}
\bibliography{ARLB}

\begin{appendices}
\section{The function $g_q(n^*)$}
We wish to find $g_q(n^*)$ so that for $\alpha_{n^*}(q)$ in equation (1),
\[
O_L(\alpha,q)=-e \cdot \alpha_{n^*}(q) \cdot \log(\alpha_{n^*}(q)) \cdot g_q(n^*) \rightarrow O(1),
\]
as $n^*$ grows. 

If we let

\[ g_q(n^*)= \left[\frac{n^*}{\log(\chi^2_{\alpha}(q)+q\cdot n^*)}\right]^{q/2}\] 
then Calculus show that the limit of $O_L(\alpha,q)$ as $n^*$ grows is equal to $\frac{e}{2}$ for any value of $q$.

\section{Predictive distribution when the likelihood is a Gamma$(\alpha,\beta$ and the prior is ScaledBeta2($p,q,b$)}
Let $X_1, \ldots, X_n$ i.i.d. Gamma($\alpha,\beta$) with $\alpha$ known. The likelihood of this sample is:

\[
f(x|\beta)= \frac{\beta^{n\alpha}}{\Gamma^n(\alpha)}\left( \prod_{i=1}^n x_i\right)^{\alpha-1} \exp\{-\beta n \bar{x}\}
\]

The prior for $\beta$ will be a ScaledBeta2($p,q,b$).

\[
\pi(\beta) = \frac{\Gamma(p+q)}{\Gamma(p)\Gamma(q)} b^q\frac{\beta^{p-1}}{(\beta+b)^{p+q}}
\]

The predictive distribution for ${\bf x}$ is, then

\[
m(\mathbf{x}) =\frac{\Gamma(p+q)}{\Gamma(p)\Gamma(q)}  \frac{\left( \prod_{i=1}^n x_i\right)^{\alpha-1}}{\Gamma^n(\alpha)} b^q I(\mathbf{x}),
\]

\noindent
where

\[
I(\mathbf{x}) = \int_0^{\infty} \frac{\beta^{(n\alpha+p)-1}}{(\beta+b)^{p+q}} \exp\{ -\beta n \bar{x}\}
\]

Making the variable change $v=n \bar{x} \beta$

\[
I(\mathbf{x}) = \frac{1}{(n\bar{x})^{n\alpha-q}} \int_0^{\infty} \frac{v^{n\alpha+p-1}}{(v+n\bar{x}b)^{p+q}} \mathrm{e}^{-v} dv
\]

Some important particular cases:

\begin{enumerate}
\item $ \mathbf{p=q=1}$
\begin{eqnarray*}
I(\boldmath{x}) & =  & \frac{1}{(n\bar{x})^{n\alpha-1}} \int_0^{\infty} \frac{v^{n\alpha}}{(v+n\bar{x}b)^{2}} \mathrm{e}^{-v} dv\\
& = & \frac{\Gamma(n\alpha)}{ (n\bar{x})^{n\alpha-1}} \left\{n(\alpha+\bar{x}b)\mathrm{e}^{n\bar{x}b}E_{n\alpha}(n\bar{x}b)-1\right\}
\end{eqnarray*}

\noindent where $E_n(x)=\int_1^{\infty} \frac{\exp(-xt)}{t^n}dt$ is the {\em Exponential Integral Function} (result calculated with Wolfram Alpha).

Therefore,
\[
m(\mathbf{x}) = \frac{\left( \prod_{i=1}^n x_i\right)^{\alpha-1}}{\Gamma^n(\alpha)}b \frac{\Gamma(n\alpha)}{ (n\bar{x})^{n\alpha-1}} \left\{n(\alpha+\bar{x}b)\mathrm{e}^{n\bar{x}b}E_{n\alpha}(n\bar{x}b)-1\right\}
\]
When $\alpha =1$ (exponential distribution)
\[
m(\mathbf{x}) =  \frac{b \Gamma(n)}{ (n\bar{x})^{n-1}} \left\{n(\bar{x}b+1)\mathrm{e}^{n\bar{x}b}E_{n}(n\bar{x}b)-1\right\}
\]

\item $\mathbf{p=q=\frac{1}{2}}$
\[ 
I(\mathbf{x})=\mathrm{e}^{n\bar{x}b} \Gamma(n\alpha+\frac{1}{2})E_{n\alpha+\frac{1}{2}}(n\bar{x}b)
\]

\noindent and

\[
m(\mathbf{x})=\frac{\sqrt{b}}{\pi} \frac{\Gamma(n\alpha+\frac{1}{2})}{\Gamma^n(\alpha)}\left(\prod_{i=1}^n x_i\right)^{\alpha-1} \mathrm{e}^{n\bar{x}b}E_{n\alpha+\frac{1}{2}}(n\bar{x}b)
\]

\end{enumerate}
\end{appendices}
\end{document}